\documentclass[letterpaper, 10 pt, conference]{ieeeconf}  % Comment this line out
\IEEEoverridecommandlockouts                              % This command is only
\usepackage{graphics} % for pdf, bitmapped graphics files
\usepackage{amsmath} % assumes amsmath package installed
\usepackage{amssymb}  % assumes amsmath package installed
\usepackage{tikz}
\usetikzlibrary{arrows, shapes}
\newcommand{\Hb}{{\bf H}}
\newcommand{\ww}{{\bf w}}

\usepackage{graphicx,url}
\usepackage{bbold,dsfont,color}
\newcommand{\T}{\mathsf{T}}
\renewcommand{\b}{{ b}}

\newcommand{\bb}{{\bf b}}
\renewcommand{\aa}{{\bf a}}
\renewcommand{\u}{{  u}}
\newcommand{\uu}{{\bf  u}}

\newcommand{\X}{{\bf X}}
\newcommand{\Z}{{\bf Z}}

\newcommand{\yy}{{\bf y}}
\newcommand{\y}{{ y}}
\renewcommand{\a}{a}
\newcommand{\eg}{\textit{e.g.,~}}
\newcommand{\ie}{\textit{i.e.,~}}
\def\subjto{{\mbox{subj. to}}}

 \newtheorem{thm}{Theorem}
\newtheorem{df}{Definition}
\newtheorem{cor}[thm]{Corollary}

\title{\LARGE \bf
Blind Identification of ARX Models with Piecewise Constant Inputs
}

\author{Henrik Ohlsson, Lillian Ratliff, Roy Dong and S. Shankar Sastry% <-this % stops a space
\thanks{The work presented is supported by the NSF
Graduate Research Fellowship under grant DGE 1106400, NSF
CPS:Large:ActionWebs award number 0931843, TRUST (Team for Research in
Ubiquitous Secure Technology) which receives support from NSF (award
number CCF-0424422), and FORCES (Foundations Of Resilient
CybEr-physical Systems), the Swedish Research
  Council in the Linnaeus center CADICS, the European Research Council
   under the advanced grant LEARN, contract 267381, a postdoctoral grant from the Sweden-America
   Foundation, donated by ASEA's Fellowship Fund, and  by a postdoctoral
   grant from the Swedish Research Council.}% <-this % stops a space
\thanks{Ohlsson, Ratliff, Dong, and Sastry are with the Department of Electrical Engineering and Computer  Sciences, University of California, Berkeley, CA, USA. Ohlsson is also with the
Division of Automatic Control, Department of Electrical Engineering, Link\"oping University, Sweden.
{\tt\small  ohlsson@eecs.berkeley.edu.}}}%

\begin{document}

\maketitle
\thispagestyle{empty}
\pagestyle{empty}

%%%%%%%%%%%%%%%%%%%%%%%%%%%%%%%%%%%%%%%%%%%%%%%%%%%%%%%%%%%%%%%%%%%%%%%%%%%%%%%%
\begin{abstract}

Blind system identification  is known to be a hard ill-posed problem and
without further assumptions, no unique solution is at hand. In this contribution, we are concerned with the task of identifying an ARX model
from only output measurements. Driven by the task of
identifying systems that are turned on and off at unknown times, we
seek a piecewise constant input and a corresponding ARX model
% that feed with these inputs, 
which approximates the measured outputs. 
We phrase this as a
rank minimization problem and present a relaxed convex formulation to approximate
its solution. 
The proposed method was developed to model power consumption of electrical
appliances  and is now a part of
a bigger energy disaggregation framework. Code will be made available online.

% It is well known that at least
%as many measurements are needed as the number of parameters for a
%unique solution. But in 
%blind ARX identification, we are trying to identify in principal both the
%system input and the ARX parameters
%simultaneously. Since  the input alone
%has as many unknowns as the number of measurements, this seems like an
%impossible task and in fact, not knowing more, this the task is impossible.

%In this paper, we assume we know a bit more. We assume that the sought
%input is piecewise constant. This assumption may seem artificial.  
\end{abstract}

%%%%%%%%%%%%%%%%%%%%%%%%%%%%%%%%%%%%%%%%%%%%%%%%%%%%%%%%%%%%%%%%%%%%%%%%%%%%%%%%
\section{Introduction}
% Consider the linear regression model
% \begin{equation}\label{eq:LR}
% y(t) = \bvarphi^\T(t) \btheta \in \Re, \quad \btheta \in \Re^n.
% \end{equation}
% Estimation of this type of model is probably the most common task in
% system identification and a very well studied problem, see for instance \cite{Ljung:99}. It is
% well know that ARX-models 

Consider an \textit{auto-regressive exogenous input} (ARX) model 
\begin{align}\nonumber
y(t)-&\a_1 y(t-1)  %+\a_2 y(t-2) 
- \cdots - \a_{n_a} y(t-n_a) \\=&
\b_1 u(t-n_k - 1) %+\b_2 u(t-n_k - 2)
+\cdots +\b_{n_b} u(t-n_k - n_b) 
%y(t)=&
%\sum_{k=1}^n \g_k u(t - k) 
\end{align} with input $u \in \Re$ and output $y \in \Re$.
Estimation of this type of model is probably the most common task in
system identification and a very well studied problem, see for instance \cite{Ljung:99}. 
The common
setting is that $\{(y(t),u(t))\}_{t=1}^N$  is given and the summed
residuals{\small
\begin{equation*}
\sum_{t=n}^N \left (y(t) - \sum_{k_1=1}^{n_b}
\b_{k_1} u(t-k_1-n_k) - \sum_{k_2=1}^{n_a}
\a_{k_2} y(t-k_2) \right )^2
\end{equation*}}
where $n=\max( n_a, n_k+n_b)+1$, is minimized to obtain an estimate for $\a_1,\dots, \a_{n_a},\b_1,\dots, \b_{n_b}$. This estimate is often
referred to as the \textit{least squares} (LS) estimate.
% It is well
%know that the \textit{finite impulse response} (FIR) model is a
%contained in the class of ARX modells and lately it has 
%been observed that clever regularizations can improve the estimates
%for FIR models,
%see for instance \cite{Pillonetto:10b,Pillonetto:10a,Ohlssonetal:11k}.

In this paper we study the more complicated problem of
estimating an ARX model from solely outputs $\{y(t)\}_{t=1}^N$. This is
an ill-posed problem and it is easy to see that under no further
assumptions, it would be impossible to uniquely determine  $\a_1,\dots, \a_{n_a},\b_1,\dots, \b_{n_b}$.

We will in this contribution study this problem under the assumption
that the input is piecewise constant. This is a rather natural
assumption and a problem faced in many identification
problems. Consider \eg the  modeling of an
electrical appliance where the power consumption is monitored while
the appliance is turned on and off. The exact time for when the
appliance was turned on and off is not known and neither is the
amplitude of the ``input''.  

It should be noticed that the assumption of a piecewise constant input
is not enough to uniquely determine the input or the ARX
model. Specifically, we will not be able to decide the input or the
ARX coefficients $\b_1,\dots, \b_{n_b}$  more than up to a multiplicative scalar. However, for
many applications this is sufficient, as we will illustrate in the numerical section. 

The task of identifying a model from only outputs is in system
identification referred to as \textit{blind system identification} (BSI). It is known to be
a difficult problem and in general ill-posed. 

  %%%%%%%%%%%%%%%%%%%%%%%%%%%%%%%%%%%%%%%%%%%%%%%%%%%%%%%%%%%%%%%%%%%%%%%%%%%%%%%%
\section{Background}
\label{sec:background}
Our work is motivated by blind system identification  which is a fundamental signal processing tool used for identifying a system using only observations of the systems output. Formally, given the output signal of a system, BSI serves as a tool for estimating the unknown inputs and system model~\cite{abed1997:jk}. 
\tikzstyle{int}=[draw, minimum size=3em]
\tikzstyle{init} = [pin edge={to-,thin,black}]
\tikzstyle{sum} = [draw, circle, node distance=2cm]
\begin{figure}[ht]
  \begin{center}
    \begin{tikzpicture}[node distance=2.5cm,auto,>=latex']
    \node [int] (a) {$H$};
    \node [sum, pin={[init]above:$w$}, right of=a] (sum) {};
    \node (b) [left of=a,node distance=2cm, coordinate] {a};
      \node [coordinate] (c) [right of=a, node distance=2cm]{};
      \node [coordinate] (end) [right of=sum, node distance=2cm]{};
    \path[->] (b) edge node {$u$} (a);
    \path[->] (a) edge node {} (sum);
    \path[->] (sum) edge node {$y$} (end);
\end{tikzpicture}
      \end{center}
  \caption{Input-Output Model of a System}
  \label{fig:IOmode}
\end{figure}
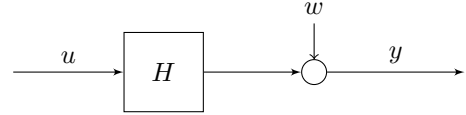

Consider the block diagram in Figure~\ref{fig:IOmode}. Suppose that the system we want to identify is linear. Given only $y$, BSI is used to identify the input $u$ and the system transfer function $H$. In this way, BSI is a method for solving the inverse problem of system identification without input information. Consider now a discrete linear, time-invariant system. Note that we describe the theory in this section for discrete time systems but the continuous time counterpart can be derived in a similar fashion. The output can be written as a convolution model, \ie 
\begin{equation}
  y(t)=u(t)\star h(t)+w(t) %=\sum_{j=-\infty}^{\infty} h_ju_{k-j} +w_k
  \label{eq:conv}
\end{equation}
where $\star$ is the convolution operator and $w$ is a noise term. This problem can be transfered to the frequency domain by applying the Fourier transform to get the following system:
\begin{equation}
  Y(s)=H(s) U(s)+W(s)
  \label{eq:Freq}
\end{equation}
An alternative name for BSI when the system is linear, time-invariant is blind deconvolution. For a known input $u(k)$, a deconvolution process can be applied to $y_i(k)$ to approximate $h_i(k)$. For instance, using a pseudo-inverse filter which is an approximation of the Weiner filter, the result of deconvolution gives
\begin{equation}
  U(s)\approx \frac{Y_i(s)H^\dagger(s)}{|H(s)|^2+C}
  \label{eq:weiner}
\end{equation}
where $(\cdot)^\dagger$ denotes the pseudo-inverse and $C$ is a constant chosen based on heuristics and serves to prevent amplification of noise~\cite{caron:2004:ji}. However, we are interested in solving the problem with unknown inputs. When the input is unknown, usually partial information about the statistical properties of $\{u(t)\}$ is required in order to obtain a good approximation of the output $\{y(t)\}$. Further, how the partial information is used in the identification problem plays an important role in the quality of the solution~\cite{li1993:hy}.

Typically, system identification requires information on the input and the output of a system in order for the problem to be well-posed and to reconstruct the system itself. However, in many applications, \eg data communications, speech recognition, image restoration and seismic signal processing, this information is not readily available. Broadly speaking, in all these application areas we can describe the identification problem using the following abstract formulation.

Suppose there is a signal that is transmitted through a `channel' that can be described using a linear, time-invariant model with a single input and $p$ outputs. The input to the system $\{u(t)\}$ results in $N$ of output sequences $\{y_1(t)\}, \ldots, \{y_p(t)\}$. Let $\{h_1(t)\}, \ldots, \{h_p(t)\}$ denote the \textit{finite impulse responses} (FIR's) which are of order $K$. As we noted above, a linear, time-invariant system of this type can be described using the convolution model in Equation~\eqref{eq:conv} for each $i\in \{1, \ldots, p\}$, \ie
\begin{equation}
  y_i(t)=u(t)\star h_i(t)+w_i(t).
  \label{eq:ideconv}
\end{equation}
This model can be concisely as
\begin{equation}
  \yy=\Hb \uu+\ww
  \label{eq:concise}
\end{equation}
where 
\begin{equation}
  \yy:=\begin{bmatrix} \yy_1^\T& \cdots& \yy_p^\T\end{bmatrix}^\T
  \label{eq:ybold}
\end{equation}
with each $\yy_i=[y_i(1)\; \cdots \; y_i(N)]^\T$, and
\begin{equation}
  \uu=\begin{bmatrix} u(-K)&u(-K+1)&\cdots & u(N-1)\end{bmatrix}^\T.
  \label{eq:uuconv}
\end{equation}
The matrix $\Hb$ takes the form
\begin{equation}
  \Hb:=\begin{bmatrix} \Hb_1^\T & \cdots & \Hb_p^\T \end{bmatrix}^\T
  \label{eq:Hb}
\end{equation}
where each $\Hb_i$ is an $N\times (N+K)$ filtering matrix given by
\begin{equation}
  \Hb_i=\begin{bmatrix} h_i(K)& \cdots& h_i(0) & \cdots & 0\\
    %0 & h_i(K) & \cdots & h_i(0) & \cdots & 0\\
    \vdots & \ddots & & \ddots & \vdots\\
   0 & \cdots & h_i(K) & \cdots & h_i(0)
  \end{bmatrix}.
  \label{eq:filtermatrix}
\end{equation}
Now that the system has been written in this form, we can formulate
the BSI or blind deconvolution problem and ask when it has a
well-defined solution. If a system identification problem is
well-posed, then all the unknown parameters can be uniquely determined
given the data. Given $\yy$ and $\ww\equiv 0$, then we can only hope
to solve the system in Equation \eqref{eq:concise} for unique $\uu$
and $\Hb$ up to a scalar~\cite{abed1997:jk}. In this case we call the
system identifiable. Necessary and sufficient conditions for
identifiability are given in~\cite{hua1996:fa} and summarized in
\cite{abed1997:jk}.

There are a number of methods for estimating either the input $\uu$ or
the system function $\Hb$. Once either the input or the system matrix
has been estimated, the other can be calculated using the
estimate. The application typically determines whether a direct
estimation of input or system matrix should be done. For instance, in
communication applications the input carries the \emph{information}
and as such direct estimation should be used for the input and the
system matrix should be calculated after. The input $\uu$ can be
estimated using the following methods: \textit{input subspace} (IS) method,
\textit{mutually referenced equalizers} (MRE), or \textit{linear prediction} (LP) method
(see~\cite{abed1997:jk, gesbert1997:lk, Makhoul:1975:lj}). The system
matrix can be directly estimated using the \textit{maximum likelihood} (ML)
method (see, for instance,~\cite{Talwar:1994:kl}) and the subspace
method (see \cite{abed:1997:fi}). We remark that in the above
formulation we have considered only FIR models. These tend to be
sufficient in practice considering that infinite impulse responses can
be approximated by FIR's and modeling with FIR's results in problem
formulations that have tractable solutions.

In this paper we are concerned specifically with estimating an ARX model from only output observations and we formulate the problem using the BSI framework.
%%%%%%%%%%%%%%%%%%%%%%%%%%%%%%%%%%%%%%%%%%%%%%%%%%%%%%%%%%%%%%%%%%%%%%%%%%%%%%%%
\section{Problem Formulation}
Given $\{y(t)\}_{t=1}^N \in \Re$ and a bound for the noise $\epsilon$, find an estimate for  $\a_1,\dots,
\a_{n_a},\b_1,\dots, \b_{n_b} \in \Re$ and an over time piecewise constant  $ 
u(t) \in \Re,t=1,\dots,N,$ such that  
\begin{align*}\nonumber
y(t)-&\a_1 y(t-1) 
- \cdots - \a_{n_a} y(t-n_a) \\=&
\b_1 u(t-n_k - 1) 
+\cdots +\b_{n_b} u(t-n_k - n_b)+w(t),
\end{align*} for  $t=n,\dots,N$, where $n=\max( n_a, n_k+n_b)+1$, and
\begin{equation}
 |w(t)| \leq \epsilon,\quad t=n,\dots,N.
\end{equation}
We will for simplicity assume that $n_a, n_b, n_k,$ are
known. To make the problem well posed, we will seek the piecewise
constant input with the least amount of changes. Other choices have
been studied for the related problem of blind deconvolution, see
\cite{Ahmed:12} for a solution where the signals to be recovered are
assumed to be in some known subspaces.

%%%%%%%%%%%%%%%%%%%%%%%%%%%%%%%%%%%%%%%%%%%%%%%%%%%%%%%%%%%%%%%%%%%%%%%%%%%%%%%%
\section{Notation and Assumptions}

We will use $y$ to denote the output and $u$ the input. We will for
simplicity only consider \textit{single input single output} (SISO)
systems. We will assume
that $N$ measurements of $y$ are available and stack them in the vector
$\yy$, \ie \begin{align} \yy=\begin{bmatrix} \y(1) &\dots &
   \y(N)  \end{bmatrix}^\T.\end{align}  We also introduce $\uu$, $\ww$, $\aa$
and $\bb$
 as
\begin{align} %\\
\uu=&\begin{bmatrix} u(1) &\dots &
    u(N)  \end{bmatrix}^\T,\\
 \ww=&\begin{bmatrix} w(1) &\dots &
   w(N) \end{bmatrix}^\T,\\ 
\aa=&\begin{bmatrix} \a_1 &\dots &
    \a_{n_a}  \end{bmatrix}^\T,\\
\bb=&\begin{bmatrix} \b_1 &\dots &
    \b_{n_a}  \end{bmatrix}^\T.\end{align} We will use $\yy(i)$ to
denote the $i$th element of $\yy$. To pick out a subvector of $\yy$
consisting of the $i$th to the $j$th element we will use the notation
$\yy(i:j)$ and similarly for picking out a subvector of $\uu$, $\aa$
and $\bb$. To pick out a submatrix consisting of the $i$th to the
$j$th rows of $\X$ we use the notation $\X(i:j,:)$.

We will use normal font to represent scalars and bold for vectors and
matrices. 
%$|\cdot|$ represents the absolute value for scalars, vectors
%and matrices and returns the number of elements of a set if the
%argument is a set.  
 $\|\cdot \|_0$ is the zero norm which returns the
number of nonzero elements of its argument and $\|\cdot\|_p$ the
$p$-norm defined as $\|{\yy}\|_p \triangleq  \sqrt[p]{\sum_i | {\yy}(i)|^p
}$.  
%For
%matrixes,  $\|\cdot\|_0$ returns the
%number of nonzero elements and $\|\cdot\|_p$ is defined as
%$ {\X} \triangleq  \sqrt[p]{\sum_{i,j}|  \X(i,j)|^p
%}$, where $\X(i,j)$ picks out the $i,j$th element of
%$\X$.  
$\|\X\|_{i,j}$ is used to denote the combination of the
$i$-norm with the $j$-norm. The $i$-norm is applied to each 
row of $\X$ and  the $j$-norm on the resulting vector. 
 %$\X^*$ is denoting the complex conjugate
%transpose of $\X$. 
%We let $\mathds{1}_{m \times n}$ denote a $m
%\times n$ matrix of ones, $I$ the identity matrix, $\Re$ the set of the real numbers, 
%and $\mathbb{Z}$ be the set of integers. $\Re \{\cdot \}$ returns
%the real part of its argument.
We will use $\Delta u$ to denote the $(N-1) \times 1$ row vector made
up of consecutive differences of $u$'s,
\begin{align*}\Delta u= &
\uu(1:N-1) -\uu(2:N) \\
=&\begin{bmatrix}
  u(1)-u(2) & \cdots & u(N-1)-u(N) \end{bmatrix}.
\end{align*}

%%%%%%%%%%%%%%%%%%%%%%%%%%%%%%%%%%%%%%%%%%%%%%%%%%%%%%%%%%%%%%%%%%%%%%%%%%%%%%%%
\section{Blind Identification using Lifting}

% If we define 
%\begin{equation}
%\Delta u= \begin{bmatrix}
 % \|u(1)-u(2)\|_2 & \cdots & \|u(N-1)-u(N)\|_2 \end{bmatrix}
%\end{equation}
We can formulate the problem of finding the input that changes most
infrequently and the ARX coefficients as the non-convex combinatorial problem
{\small \begin{subequations}\label{eq:probform}
\begin{align}\label{eq:probform1}
\min_{\begin{array}{cc}u(t),w(t)\,t=1,\dots,N,\\ \a_1,\dots,
\a_{n_a}, \b_1,\dots, \b_{n_b} \end{array}}  &\quad   \|\Delta u \|_0,
\\  \hspace{0cm} \label{eq:probform2}
  \subjto   \quad y(t)-\a_1 y(t-1) & %+\a_2 y(t-2) 
- \cdots - \a_{n_a} y(t-n_a) \\=
\b_1 u(t-n_k - 1) 
+\cdots  &+\b_{n_b} u(t-n_k - n_b)+w(t), \\  \quad |w(t)&|\leq \epsilon,\quad  t=n,\dots,N,
\end{align}\end{subequations}}
with the zero-norm counting the number of nonzero
elements of $\Delta u$. 
%Note that we in problem \eqref{eq:probform}
%seek the sparsest $\Delta u$. A natural set of questions are then: How can we
%be sure that the true $u$  corresponds to the sparsest $\Delta u$? Why
%could not the true solution correspond to the second sparsest
%$\Delta u$?. In other words, what is the motivation for seeking the
%sparsest $\Delta u$ satisfying the measurements? As it turns out, if the true
%$u$ is has few enough changes, then  \TODO{..}
Note that the combinatorial nature of the zero-norm alone
makes 
\eqref{eq:probform} difficult to solve. In addition    $\{
\a_k\}_{k=1}^{n_a}$, $\{ \b_k\}_{k=1}^{n_b}$,  $ \{
w(t)\}_{t=1}^N $   and  $ \{
u(t)\}_{t=1}^N $ are unknown, which makes even small
problems ($N$  small) difficult to solve.

Introduce 
$\X = \uu \bb^\T \in \Re ^ {N    \times n_b}.$ 
If we assume 
that  $\|\bb\|_2  \neq 0$, the objective of
\eqref{eq:probform} can  be written as \begin{align}\nonumber \|\Delta \u
  \|_0 =&\left \| \|\bb\|_2  %\begin{bmatrix}  |u(1) -u(2)|  \\ |u(2) -u(3)| \\ \vdots \\ |u(N-1) -u(N)|      % \end{bmatrix}
\left (\uu(1:N-1)-\uu(2:N) \right)
\right \|_0  \\=& \|\X(1:N-1,:)-\X(2:N,:)  
  \|_{2,0} \end{align}
%\end{equation} 
 %and
%\begin{equation} \Delta \X = %\begin{bmatrix} u(1)-u(2)\\
  %  u(2)-u(3) \\ \vdots \\ u(N-1)-u(N) \end{bmatrix} \begin{bmatrix} \b_1 &\b_2
   % &\cdots &b_{n_b} \end{bmatrix}   
%\left (\uu(1:N-1) -\uu(2:N) \right) \gg^\T \in \Re ^ {(N-1)
 %   \times n}.\end{equation}
 %and
%\begin{equation} \y = \begin{bmatrix} y(1) &
 %   y(2) & \hdots & y(N) \end{bmatrix}^\T \in \Re ^
 % {1\times N
  %  }.\end{equation}
Problem \eqref{eq:probform}  can  now be reformulated as
\begin{subequations}\label{eq:probformm}
\begin{align}\label{eq:probformm1}
\min_{\X, \ww,\aa,\bb}  \quad   \|& \X(1:N-1,:)-\X(2:N,:)  \|_{2,0}
\\  \label{eq:probformm2}
 \subjto \quad  y(t) &=
 \sum_{k_1=1}^{n_b} \X(t-n_k-k_1,k_1) \\ \quad &+\sum_{k_2=1}^{n_a} \a_{k_2} y(t-k_2)+w(t),\\
  \quad |w(t)|&\leq \epsilon,\quad t=n,\dots,N, \\ \quad  rank(\X) &=1.
\end{align}\end{subequations} 
This problem is equivalent with \eqref{eq:probform}  in the following
sense. Assume that \eqref{eq:probformm}, has a unique solution $\X^*$,
then  $\X^*$ must satisfy
$\X^*=\uu^* (\bb^*)^\T$, with $\uu^*$ and $\bb^*$ solving
\eqref{eq:probform}. Extracting the rank 1 component of $\X^*$, using
\eg singular value decomposition, we can hence decide both $\uu^*$ and
$\bb^*$  up to a multiplicative scalar (note that we can never do
better with the information at hand, not even if we would be able to
solve \eqref{eq:probform}). The estimate of $\aa$ will be
identical for both problems. 

The technique of introducing
the matrix $\X$ to avoid products between $\uu$ and $\bb$ is well
known in optimization and referred to as \textit{lifting} \cite{shor87,Lovász91,Nesterov98,Goemans:1995}.

Problem \eqref{eq:probformm} is   combinatorial and
nonconvex and therefore not  easier to solve than
\eqref{eq:probform}. To get an optimization problem we can solve, we
 relax the zero norm with the $\ell_1$-norm and remove the rank
constraint and instead minimize the rank. Since the rank of a matrix is
not a convex function, we replace the rank with a convex
heuristic. Here we choose the nuclear norm, but other heuristics are
also available (see for instance \cite{Fazel01arank}).  We then obtain the convex
 program   {\small
\begin{subequations}\label{eq:probform3}
\begin{align}
\min_{\X, \ww,\aa,\bb}  \quad   \|\X\|_*& + \lambda  \| \X(1:N-1,:) -\X(2:N,:)  \|_{2,1}
\\  \label{eq:probformm2}
 \subjto \quad  y(t) =&
 \sum_{k_1=1}^{n_b} \X(t-n_k-k_1,k_1) \\+&\sum_{k_2=1}^{n_a} \a_{k_2}
 y(t-k_2) + w(t), \\ \quad |w(t)|\leq& \epsilon,\quad
 t=n,\dots,N,
\end{align}\end{subequations} }
which we refer to as \textit{blind identification via
  lifting} (BIL) of ARX models with piecewise constant
input. $\lambda>0$ is a design parameter that roughly decides the
tradeoff between rank of $\X$ and the number of changes in
the input. Ideally, $\lambda$ is set to some large number and then
decreased until the solution $\X$ to BIL  becomes rank one.

\section{Analysis}
In this section, we highlight some theoretical results derived for BIL. 
The analysis follows that of CS, and is inspired by derivations given in
\cite{ohlsson:13,Candes:11,Candes:06,Chai:10,Donoho:06,Candes_2008,berinde:08,bruckstein:09,Candes:2010}. 

We need the following generalization of the RIP-property.   
\begin{df}[RIP]
We will say that a linear operator $\mathcal{A} : \Re^{n_1 \times n_2}
\rightarrow \Re^{n_3}$ is $(\varepsilon,k )-RIP$ if  
\begin{equation}\label{eq:RIP}
\left | \frac{\|\mathcal{A} (\Z)\|_2^2}{\| \Z(:)\|_2^2 } -1 \right |< \varepsilon
\end{equation}
for all $n_1 \times n_2$-matrices $\Z$ satisfying 
\begin{align}\label{eq:1} 0=&\|
\Z(1,:)- \Z(2,:)  \|_{2,0} \\ \label{eq:2}  0=&\|
\Z(n_1-1,:)- \Z(n_1,:)  \|_{2,0}  \\ \label{eq:3}  0<&\|
\Z(1:n_1-1,:)- \Z(2:n_1,:)  \|_{2,0} \leq k\end{align} and $\Z\neq 0$.  $\Z(:)$ is
here used to denote the vectorization of the matrix $\Z$.
\end{df}

We can now state the following theorem:
\begin{thm}[Uniqueness]\label{thm:one}
If $ \Z$ satisfies $\bb = 
 \mathcal{A}(  \Z)$, 
\begin{equation} 
0<\|
 \Z(1:n_1-1,:)-  \Z(2:n_1,:)  \|_{2,0} \leq k
\end{equation} and
$\mathcal{A}$ is $(\varepsilon,2 k )-RIP$ with $\varepsilon<1$ then there
exist no other solutions to   $ \bb =
 \mathcal{A}(  \Z)$ satisfying \eqref{eq:1}--\eqref{eq:3}.
%\begin{equation}0<\|
%\Z(1:n_1-1,:)- \Z(2:n_1,:)  \|_{2,0} \leq k.\end{equation}
\end{thm}
\begin{proof}
Assume the contrary, \ie that there exist another solution $\tilde \Z
$
such that $\tilde \Z \neq  \Z$ and that satisfies
\eqref{eq:1}--\eqref{eq:3}. 
%\begin{equation}0<\|
%\tilde \Z(1:n_1-1,:)- \tilde \Z(2:n_1,:)  \|_{2,0} \leq k.\end{equation} 
 It is clear that \eqref{eq:1} and \eqref{eq:2} hold. In addition, \begin{align}0<\| \nonumber
\tilde \Z(1:n_1-1,:)- &\Z(1:n_1-1,:)  \\- (\tilde \Z(2:n_1,:)-& 
\Z(2:n_1,:) ))  \|_{2,0} \leq 2k.\end{align} Hence \eqref{eq:RIP} must hold for
$\tilde \Z- \Z$. But since $\mathcal{A}(\tilde \Z)=\mathcal{A}(\Z)= \bb$ we get from
\eqref{eq:RIP} that $1<\varepsilon$, which is a contradiction. We hence
have that $\Z$ is unique solution to   $ \bb =
 \mathcal{A}(  \Z)$ satisfying  \eqref{eq:1}--\eqref{eq:3}. 
%\begin{equation}0<\|
%\Z(1:n_1-1,:)- \Z(2:n_1,:)  \|_{2,0} \leq k.\end{equation}
\end{proof}

The following corollary now follows trivially.
\begin{cor}[Recoverability]\label{cor:first}
Let $\Z^*$ be the solution of 
\begin{equation}\label{eq:rel}
\begin{aligned}
\min_{\Z}  \quad   \|\Z\|_*& + \lambda  \| \Z(1:n_1-1,:) -\Z(2:n_1,:)  \|_{2,1}
\\  
 \subjto \quad  \bb =& \mathcal{A} (\Z).
\end{aligned} \end{equation} 
 If $\mathcal{ A}$
is $(\varepsilon,2 k )-RIP$ with $\varepsilon<1$, $\Z^*$ satisfies
\eqref{eq:1}--\eqref{eq:3}
%\begin{equation}0<\|
%\Z^*(1:n_1-1,:)- \Z^*(2:n_1,:)  \|_{2,0}  \leq k\end{equation} 
and $rank( \Z^*)=1$, then
$\Z^*$ is also the solution of \begin{equation}\label{eq:unrel}\begin{aligned}
\min_{\Z}  \quad &    \| \Z(1:n_1-1,:) -\Z(2:n_1,:)  \|_{2,0}
\\  
 \subjto \quad & \bb = \mathcal{A} (\Z),\quad rank(\Z)=1.
\end{aligned} \end{equation}
%The relaxed \label{eq:rel}
%problem   hence give the same $\Z$ as \eqref{eq:unrel}.
\end{cor}
\begin{proof}The corollary follows directly from
  Theorem~\ref{thm:one}.\end{proof}
It is easy to see that \eqref{eq:probform3} has the same form as
\eqref{eq:rel} and  \eqref{eq:unrel} as
\eqref{eq:probformm}. Corollary \ref{cor:first} hence provides
necessary conditions for when the relaxation, going from
\eqref{eq:probformm} to  \eqref{eq:probform3}, is tight. 

\section{Solution algorithms and software}

Many standard methods of convex optimization can be used to solve the problem
\eqref{eq:probform3}.
Systems such as CVX \cite{cvx1,cvx2} or YALMIP \cite{Yalmip}
can readily handle the nuclear norm and the sum-of-norms
regularization. For large scale problems, the \textit{alternating direction
method of multipliers} (ADMM, see \eg \cite{bert:97,boyd:11}) is an attractive choice and we have
previously shown that ADMM can be very efficient on similar
problems \cite{ohlsson:13}. Code for solving \eqref{eq:probform3} will
be made
available on \url{http://www.rt.isy.liu.se/~ohlsson/code.html}

\section{Numerical Illustrations}
\subsection{A Simple Noise Free FIR  Example}\label{ex:first}
In this example, given $\{y(t)\}_{t=1}^{30}$ and $n_a=0,n_b=3$, we illustrate
the ability to recover the FIR model used to generate
$\{y(t)\}_{t=1}^{30}$  and the correct piecewise constant input
$\{u(t)\}_{t=1}^{30}$ (up to a multiplicative scalar). The given $y$ is shown in Figure
\ref{fig:output} and the input that was used to generate $y$ in Figure
\ref{fig:input}. The true $\bb$ was $\begin{bmatrix} -7.4111  &
  -5.0782   &-3.2058 \end{bmatrix}$.  
\begin{figure}[h!]\centering
\includegraphics[width=0.9\columnwidth]{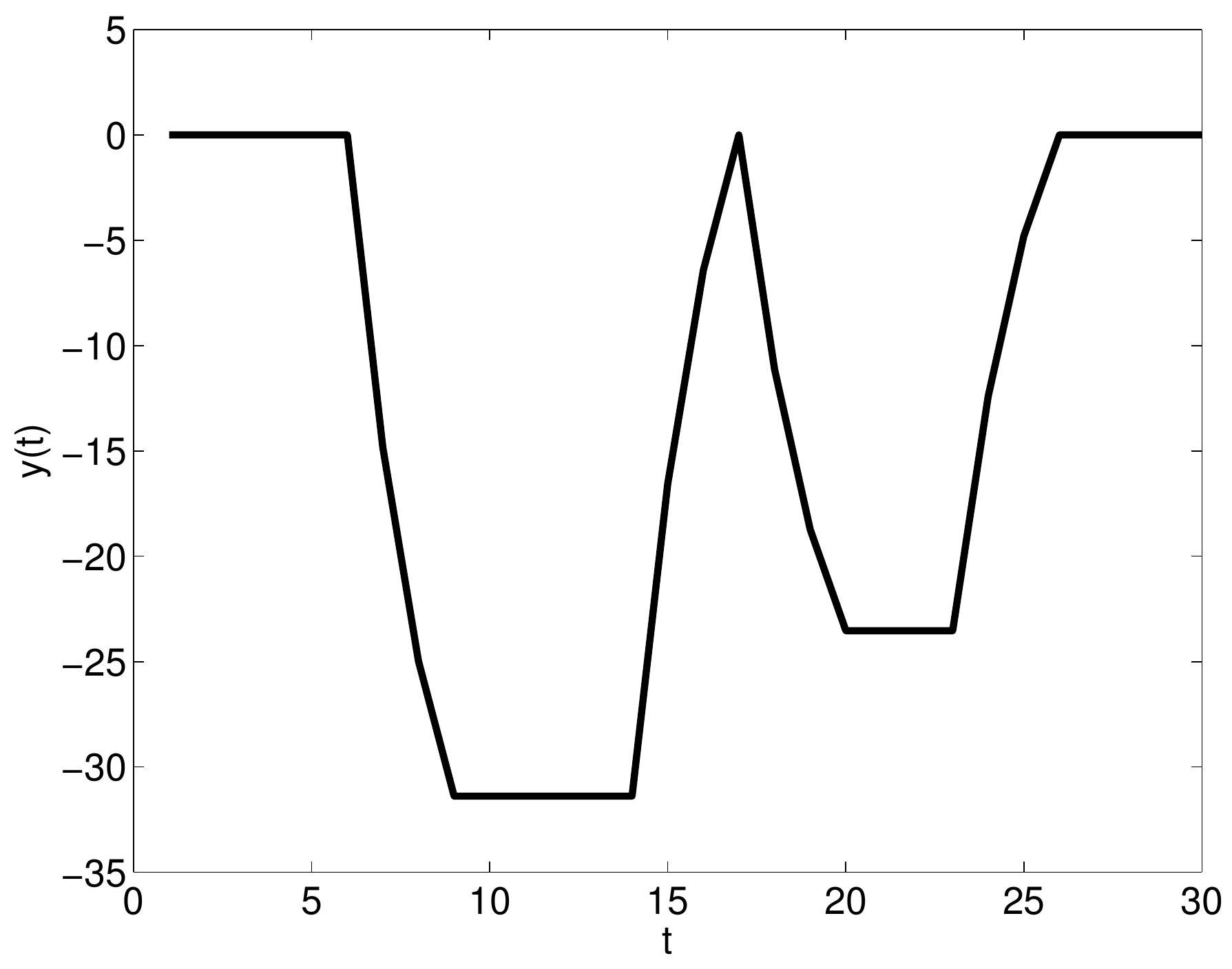}
\caption{The true output of the FIR model identified in Section \ref{ex:first}.}\label{fig:output}
\end{figure}

\begin{figure}[h!]\centering
\includegraphics[width=0.9\columnwidth]{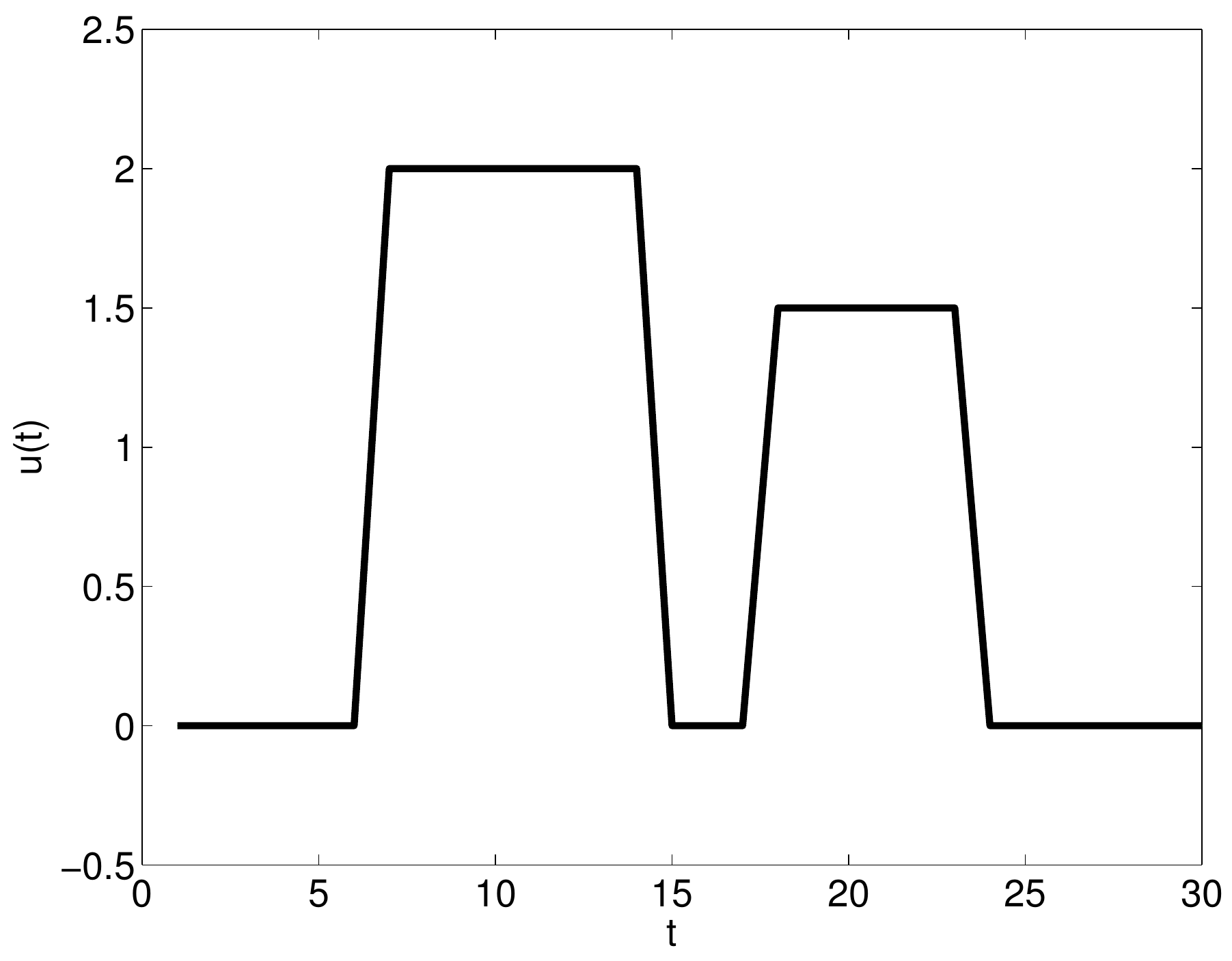}
\caption{The true input of the FIR model identified in Section \ref{ex:first}.}\label{fig:input}
\end{figure}

To recover $\{u(t)\}_{t=1}^{30}$  and $\bb$ we use BIL. $\epsilon$ was
set to $0$ and $\lambda$
was increased until  the first singular value was significantly
larger than the second singular value. $\lambda=10^4$ gave a first
singular value of 64.3 and a second singular value of $9.8 \times
10^{-6}$. The estimated input can for this $\lambda$  not be distinguished from the true
and the estimate for $\bb$ is equal to the true $\bb$ up the to the numerical precision of
the solver after rescaling. 

On this simple example, a method that first estimates the input and
then the FIR coefficients (for instance \cite{abed1997:jk,
  gesbert1997:lk, Makhoul:1975:lj}) works pretty well. In particular,
the na\"{i}ve approach of first estimating
a piecewise input by fitting a piecewise constant signal  to the output
measurements (use \eg \cite{Kimetal:09,OhlssonLB:10}) and 
 secondly estimate the FIR coefficients gave an as good result as
 BIL. 

\subsection{Identifying an ARX Model From Noisy Data}
In this example we use the same input as in the previous example but
modify the system to be
\begin{align}\nonumber
z(t)=&0.2 z(t-1)  -4.9594 u(t-1)  \\+& 6.1774 u(t-2) +   3.3930 u(t-3).
\end{align}
We also assume that there is a uniform measurement noise between $-2$
and $2$ added to the
output,
\begin{equation}
y(t)=z(t)+e(t),\quad e(t) \sim U(-2,2).
\end{equation}
Given $\{y(t)\}_{t=1}^{30}$ we now aim to find a model of the form
\begin{equation}
z(t)=a_1 z(t-1)  +b_1 u(t-1)  + b_2 u(t-2) +   b_3 u(t-3),
\end{equation}
and a piecewise constant input $\{u(t)\}_{t=1}^{30}$. The given output
sequence $\{y(t)\}_{t=1}^{30}$ is shown in Figure~\ref{fig:outputf}.  
\begin{figure}[h!]\centering
\includegraphics[width=0.9\columnwidth]{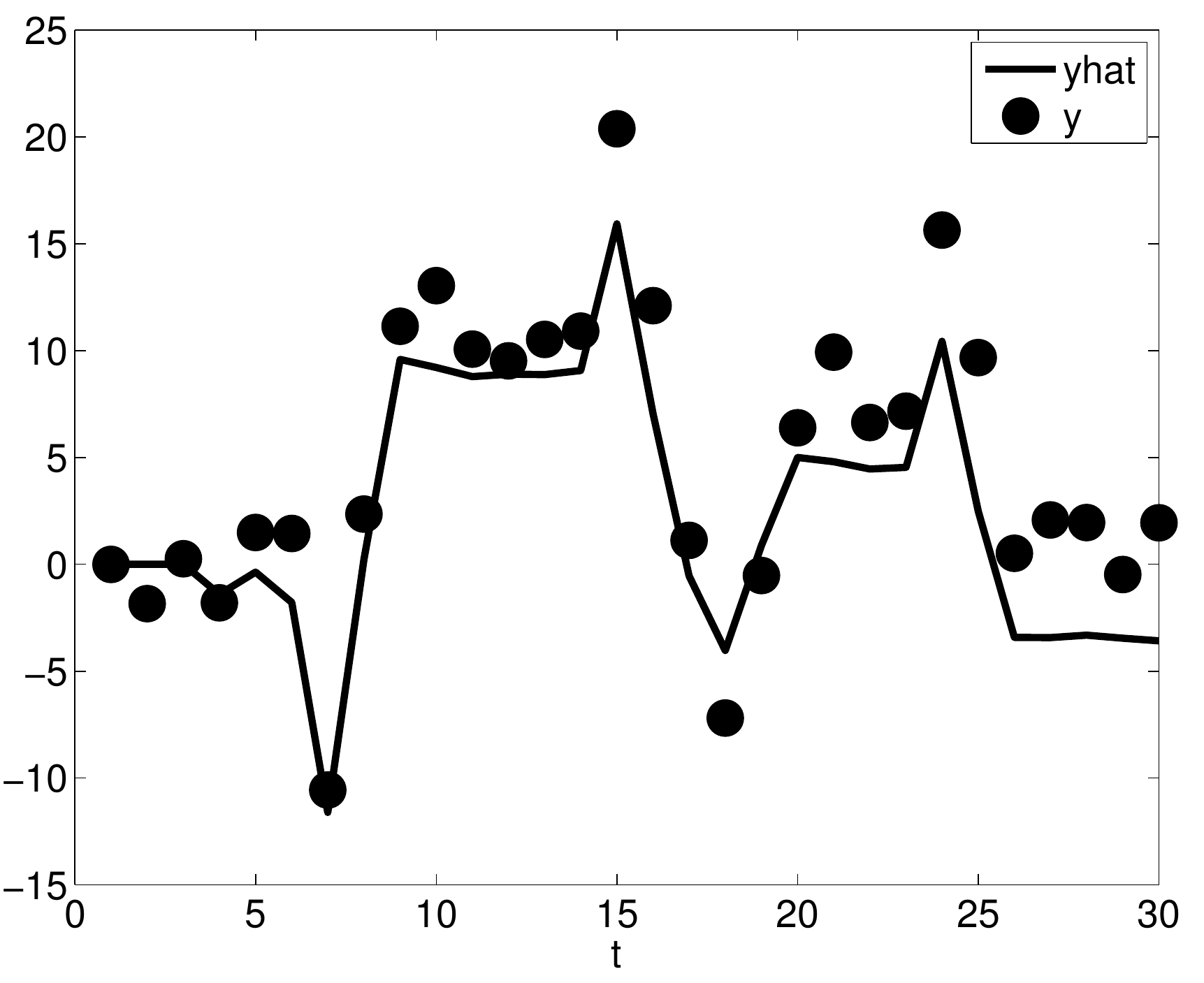}
\caption{The output measurements $\{y(t)\}_{t=1}^{30}$ given as
  filled circles and the estimated simulated output obtained by
  feeding the estimated ARX model with the estimated inputs depicted
  using solid line.}\label{fig:outputf}
\end{figure} 

If we apply
BIL with $\lambda=10^7$ and $\epsilon =2$ the input shown with solid
line in Figure~\ref{fig:inputf}
is found. The input associated with the second largest singular value
is also shown (gray thin line).  The
two largest singular values were 43  and 15. Figure~\ref{fig:inputf}
also shows the true input with dashed line. Figure~\ref{fig:outputf}
shows the output generated by driving the estimated ARX model with the input
estimate (both corresponding to the largest singular value of $\X$).
\begin{figure}[h!]\centering
\includegraphics[width=0.9\columnwidth]{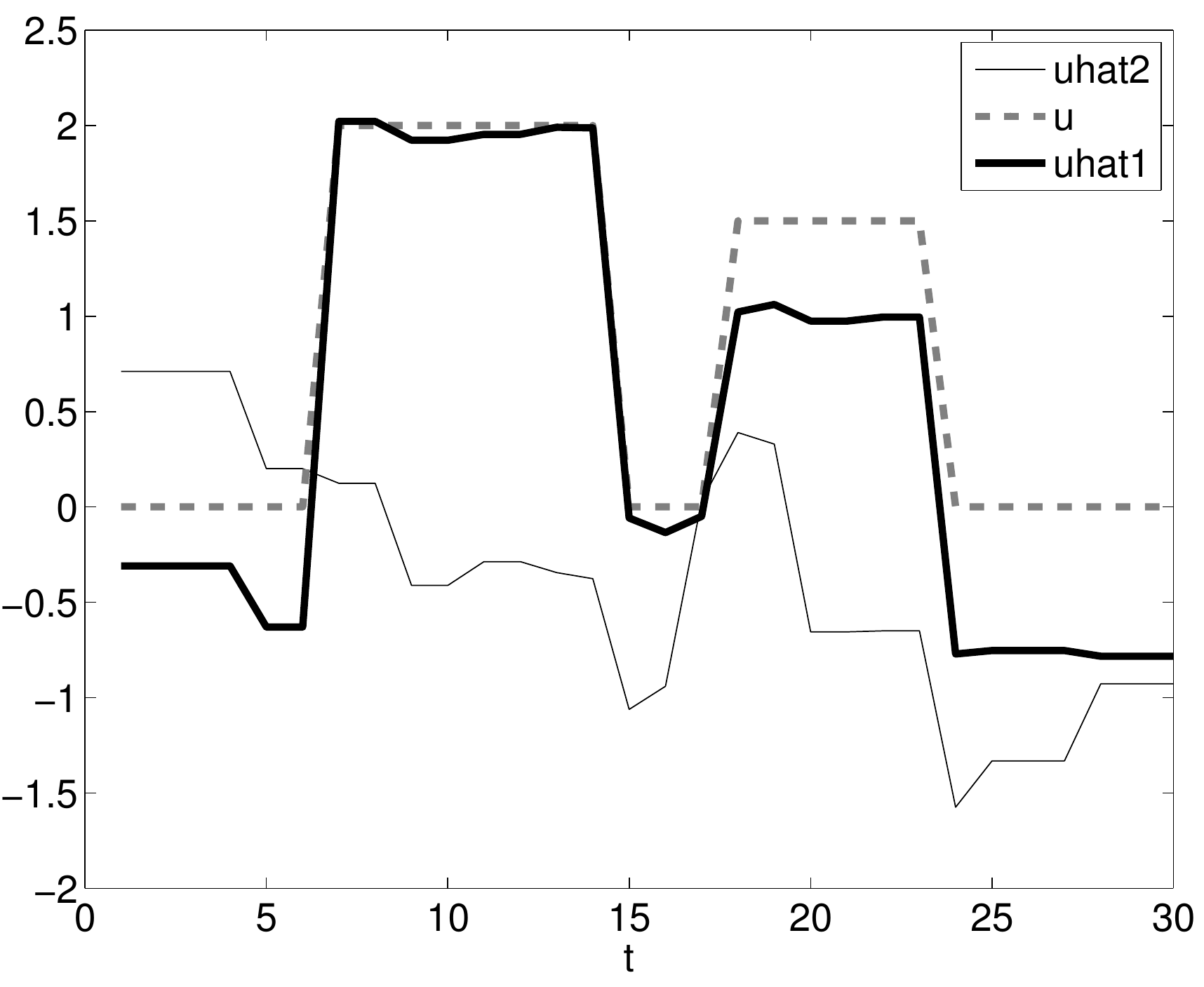}
\caption{The true input shown using dashed line, the estimate of the
  input associated with the first singular value of $\X$ with solid
  and the estimate associated with the second largest singular value
  shown with thin gray line.}\label{fig:inputf}
\end{figure}

As in Lasso \cite{Tibsharami:96} and my other $\ell_1$-regularization
problems, it is useful with a refinement step to remove bias. Simply
set $\lambda=0$ in BIL and add the constraint
\begin{equation}
\Delta \uu(i)= 0 \quad \text{ if }\quad  |\Delta \uu^*(i)| \leq \gamma,\quad
i=1,\dots, N-1,
\end{equation}
where $ \Delta \uu^* $ is the previous estimate of $\Delta \uu$ and
$\gamma \geq 0$. If we chose $\gamma=0.5$ the input shown in
Figure~\ref{fig:inputff} is the result. 
\begin{figure}[h!]\centering
\includegraphics[width=0.9\columnwidth]{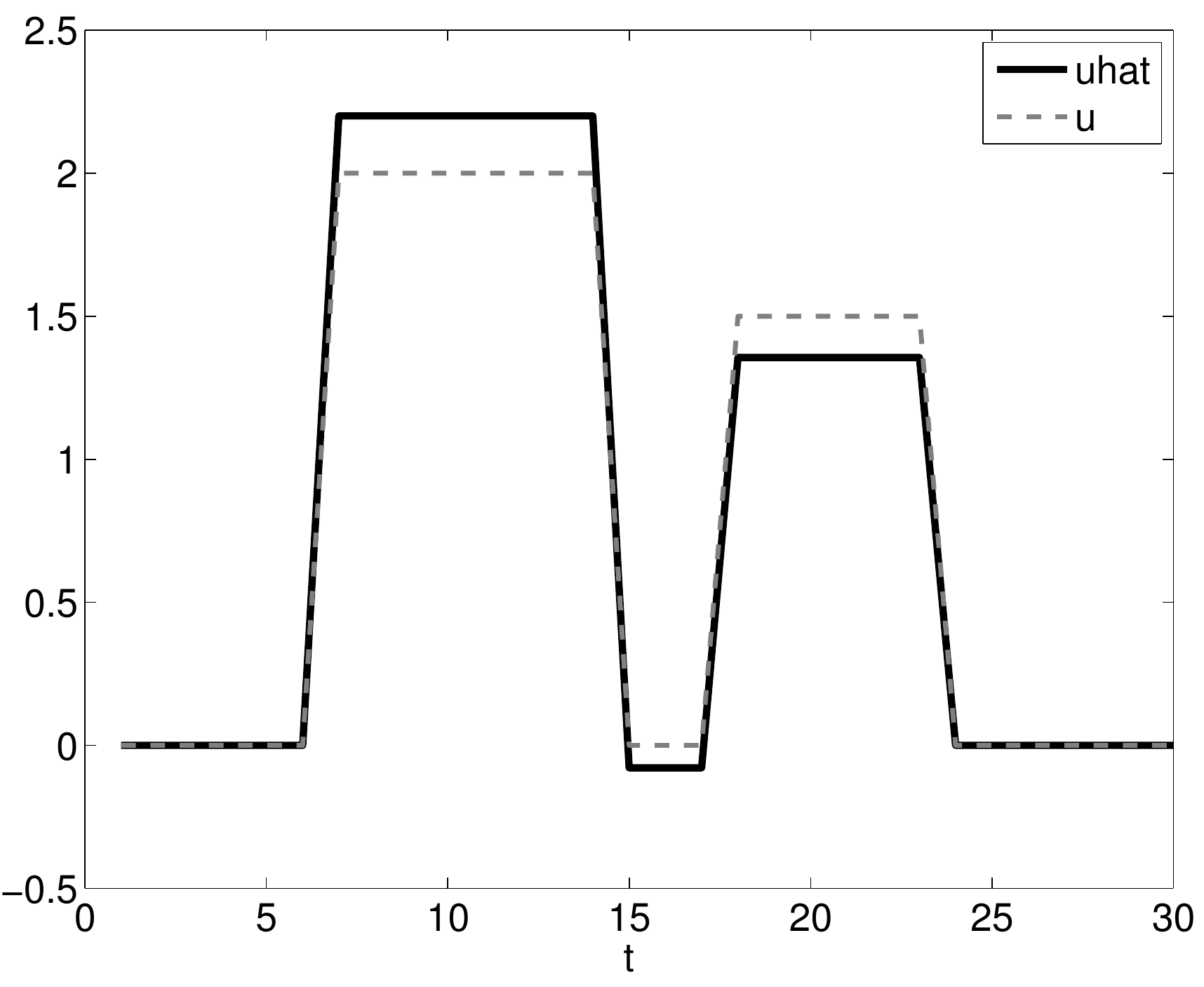}
\caption{The true input shown with dashed line and the final estimate
of the input, after the refinement step, shown with solid line. }\label{fig:inputff}
\end{figure}
The corresponding output obtained by driving the estimated ARX model
with the estimated input (both corresponding to the largest singular
value of $\X$ after the refinement step) is given
in  Figure~\ref{fig:outputff}. The
two first singular values were now 44 and 5. The estimate for $a$
was 0.2 and  $\hat b_1=-4.5594$, $\hat b_2=   5.7741$, and $\hat b_3 =
    4.0817$.

\begin{figure}[h!]\centering
\includegraphics[width=0.9\columnwidth]{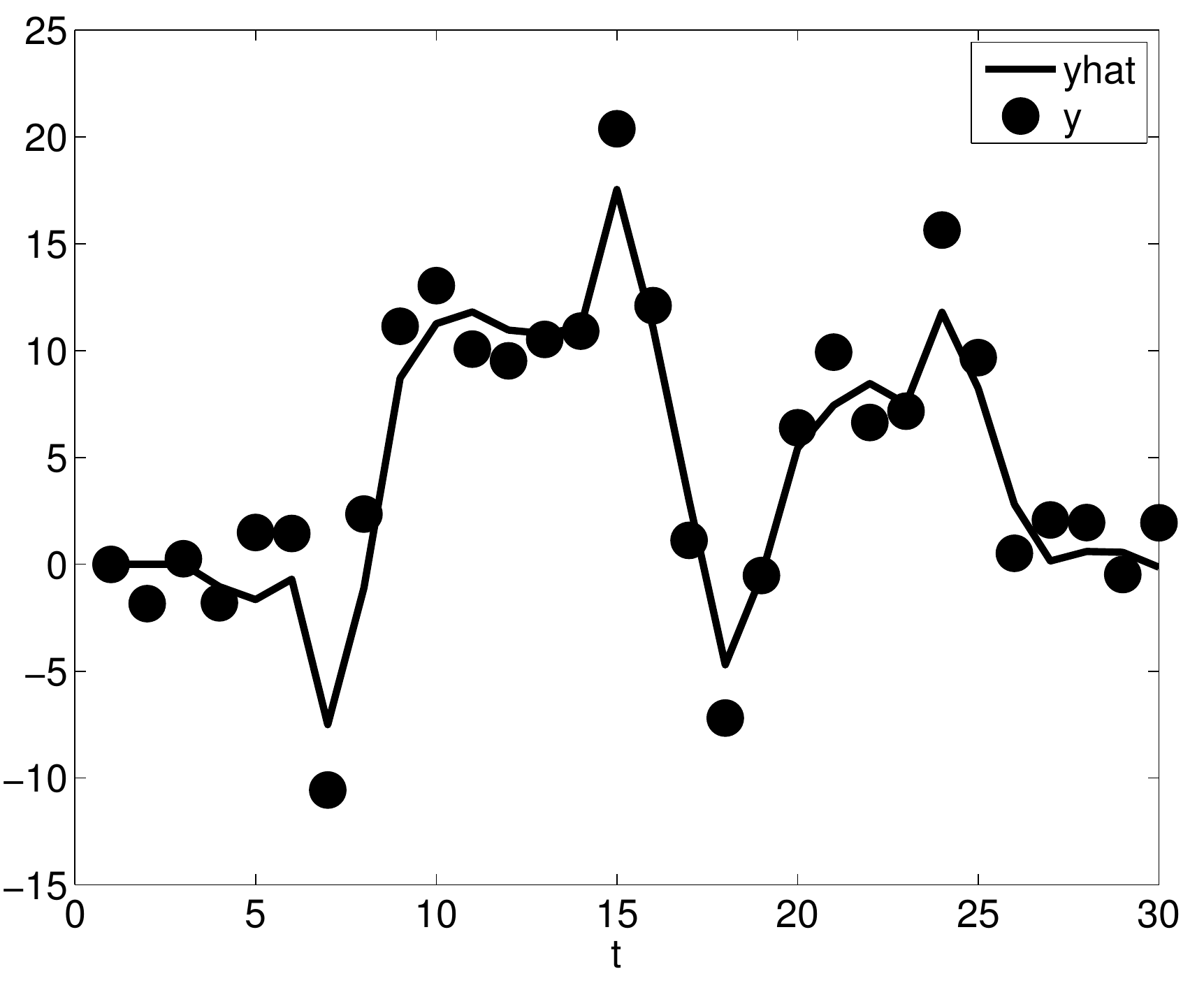}
\caption{The generated output obtained by feeding the estimated
  refined ARX model with the refined output estimates shown with solid
line. The measured noisy outputs shown with solid circles.}\label{fig:outputff}
\end{figure}

On this more challenging example, the na\"{i}ve method of first estimating
a piecewise input  and 
 secondly estimate the ARX coefficients did not give a satisfying
 result. Figure \ref{fig:altesty} shows the result of fitting a
 piecewise constant signal to the outputs and Figure \ref{fig:altestu}
 shows compares the true input with the estimated input for the na\"{i}ve method.

\begin{figure}[h!]\centering
\includegraphics[width=0.9\columnwidth]{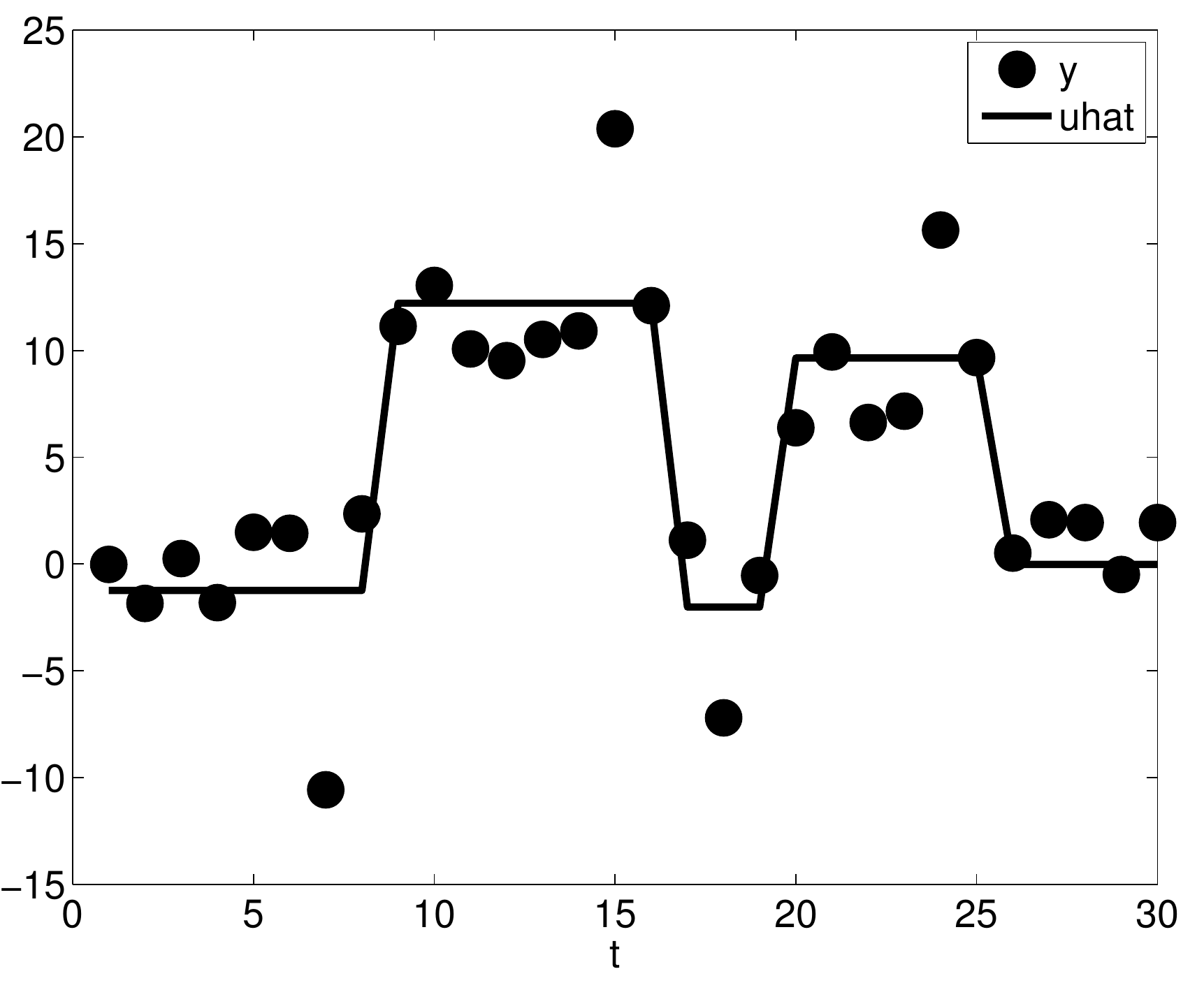}
\caption{A piecewise constant signal fitted to the measured output.}\label{fig:altesty}
\end{figure}
\begin{figure}[h!]\centering
\includegraphics[width=0.9\columnwidth]{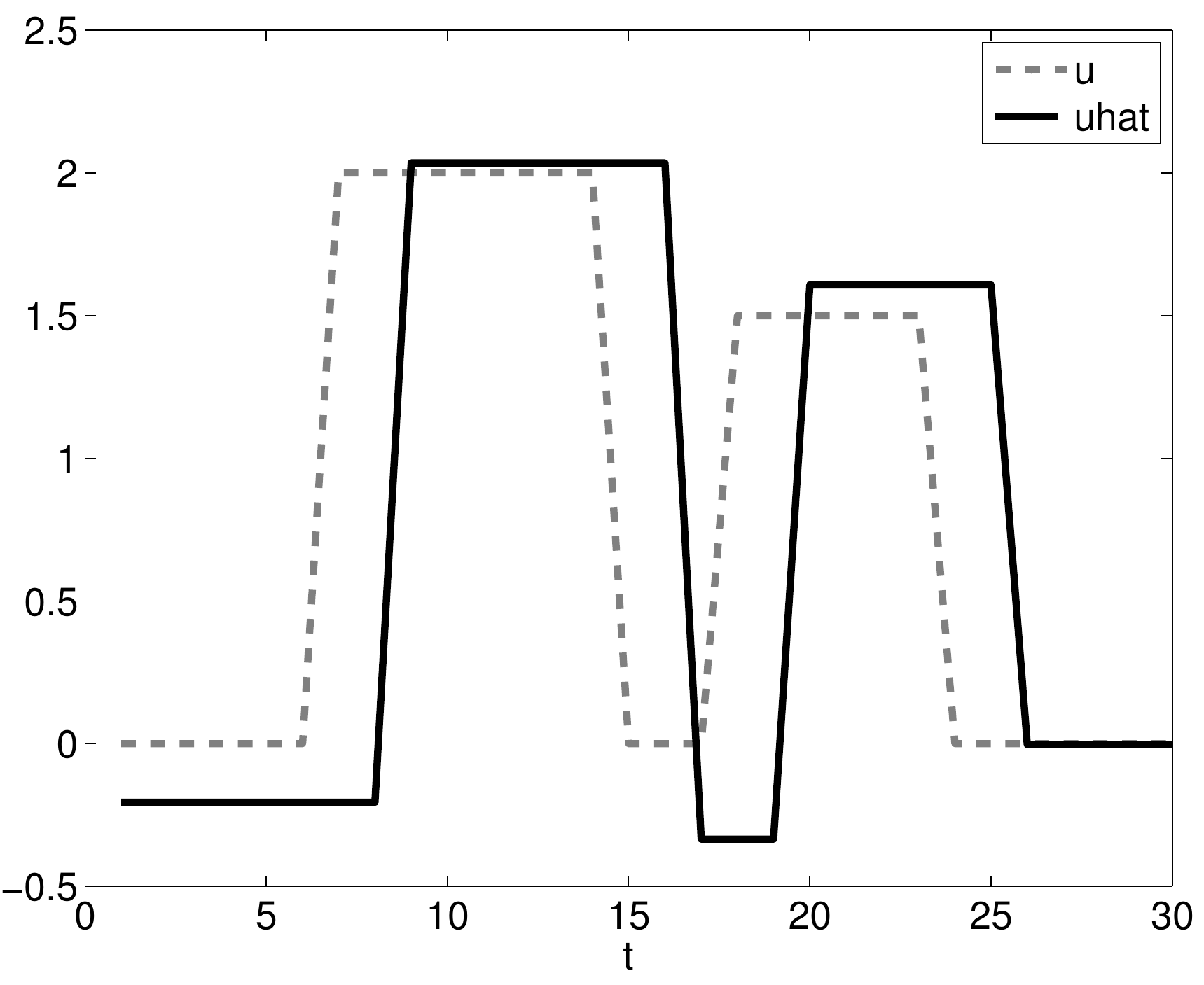}
\caption{The  true input (dashed line) and the estimated input (solid
  line) by  the na\"{i}ve method.}\label{fig:altestu}
\end{figure}

\subsection{A Real Data Example}
This example is motivated by energy disaggregation. The problem of
disaggregation refers to the problem of decomposing an aggregated
signal into its sources. As an example, the aggregated signal could be
the total energy consumed by a house. The sources would then be the
energy consumed by different appliances, \eg the toaster, HVAC,
dishwasher etc. In  \cite{Ohlssonetal:13e}, we present a disaggregation algorithm which utilizes models for individual appliances.
 To model different appliances, the power of individual
appliances was measured as they were turned on and off. Figures
\ref{fig:toaster1} and \ref{fig:toaster2} show the measured power of a
toaster as it was turned on at two different times. To estimate a model for the toaster, we need to
estimate both the input and the model at the same time. In addition,
we do not want to assume that the input is binary since many
appliances have settings that may have changed from one time to the
next, \eg the temperature setting of a toaster etc. We make the
assumption that a change in \eg the temperature of the toaster can be
modeled by different input amplitudes. It is
therefore more natural to assume that the input is piecewise constant
rather than binary.  

We chose to use $n_a=n_b=8$,  $n_k=0$ $\epsilon=0.04$ and $\lambda=10^8$. We subtracted
the total mean of both power measurements and sought  two input
sequences and a set of parameters that well approximate the two power
measurement sequences.  
This resulted in the input estimates  shown in Figures \ref{fig:toaster1u}
and \ref{fig:toaster2u}. 
\begin{figure}[h!]\centering
\includegraphics[width=0.9\columnwidth]{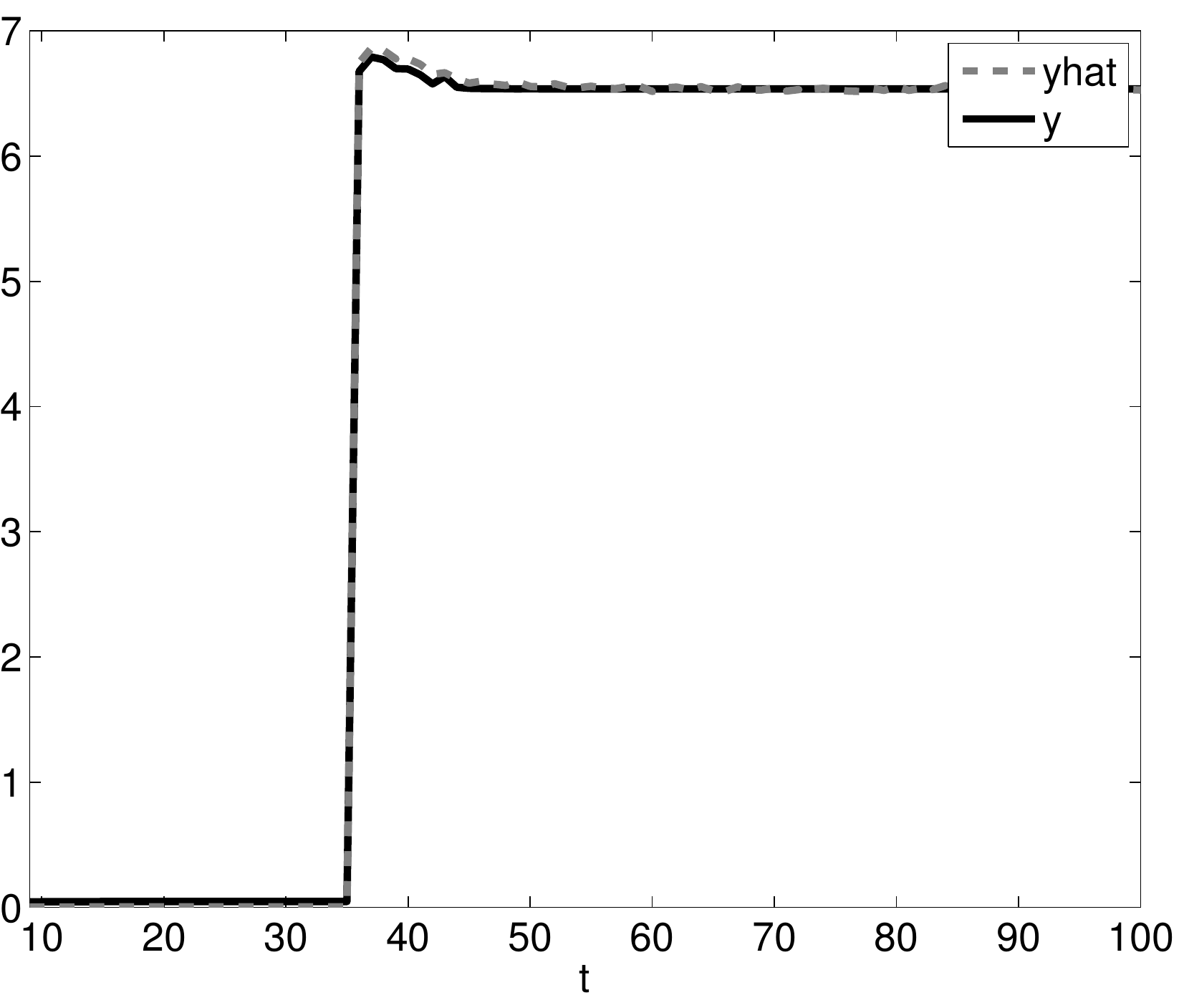}
\caption{Measured and estimated power consumption of a toaster.}\label{fig:toaster1}
\end{figure}

\begin{figure}[h!]\centering
\includegraphics[width=0.9\columnwidth]{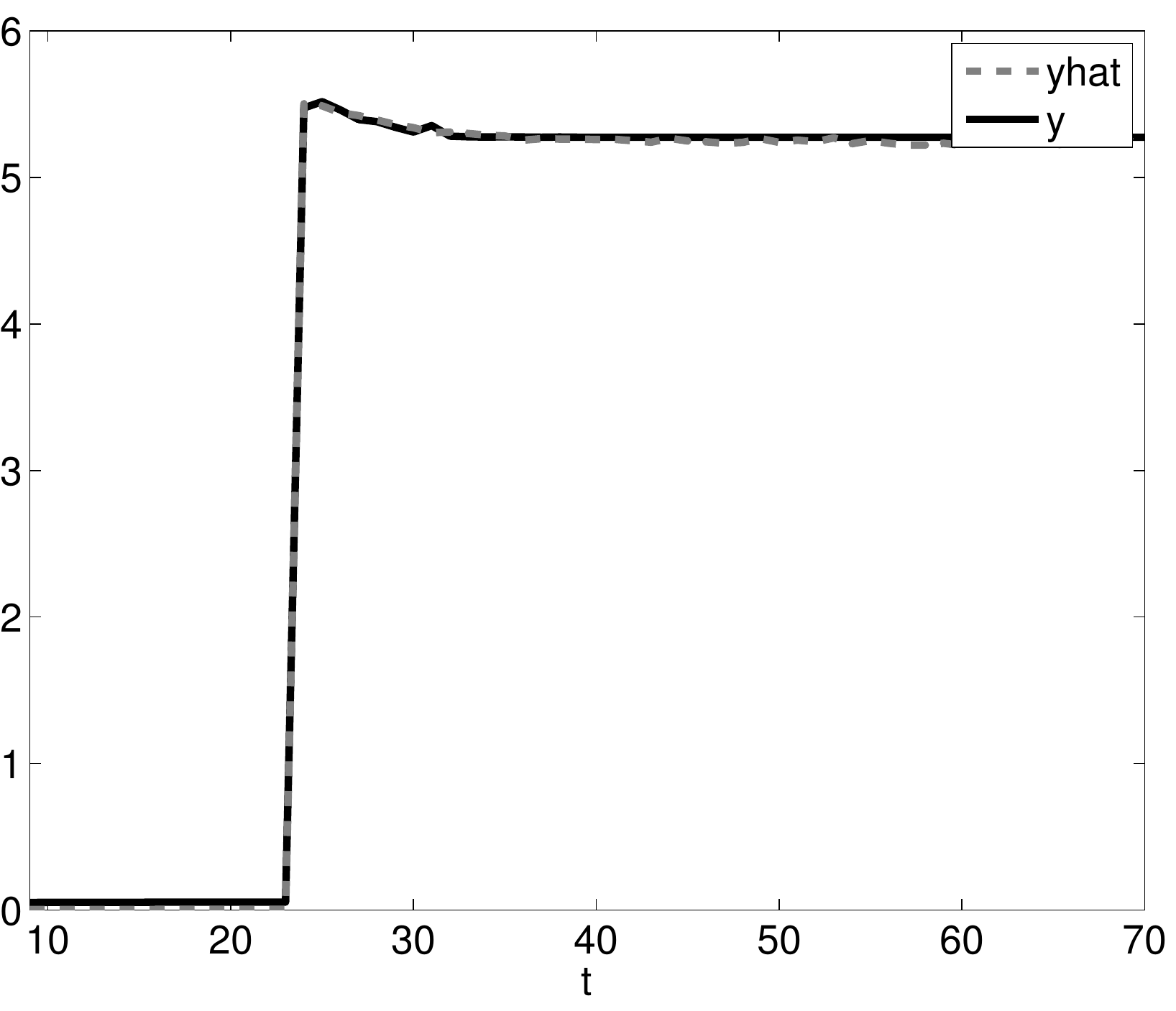}
\caption{Measured and estimated power consumption of a toaster.}\label{fig:toaster2}
\end{figure}
 The ARX parameter were computed to:
\begin{align}
\aa=  \begin{bmatrix}0.0191\\    0.0004\\   -0.0006\\    0.0098\\    0.0053\\    0.0065\\   0.0231\\   -0.0135\end{bmatrix},\quad
 \bb=\begin{bmatrix} 4.6219\\   -0.0527\\   -0.0527\\   -0.0527\\   -0.0567\\   -0.0567\\   -0.0567\\   -0.0683\end{bmatrix}.
\end{align}
Simulating the model provides the power estimates also  shown in Figures
\ref{fig:toaster1} and \ref{fig:toaster2}. The two largest eigenvalues
were 32 and 0.07. The found solution is hence very closet to being a
rank 1 matrix.

Given aggregated power measurements, we can now use
the model of the toaster and seek the a piecewise constant signal
representing the toaster being turned on and off. Since it is the
power consumption of different devices that are of interest in disaggregation, it is
not a problem that we can not identify the inputs or the ARX
parameters more than up to a multiplicative constant.  %We refer the
%reader interested in applying BIL for disaggregation to \cite{Ohlssonetal:13e}.  

\begin{figure}[h!]\centering
\includegraphics[width=0.9\columnwidth]{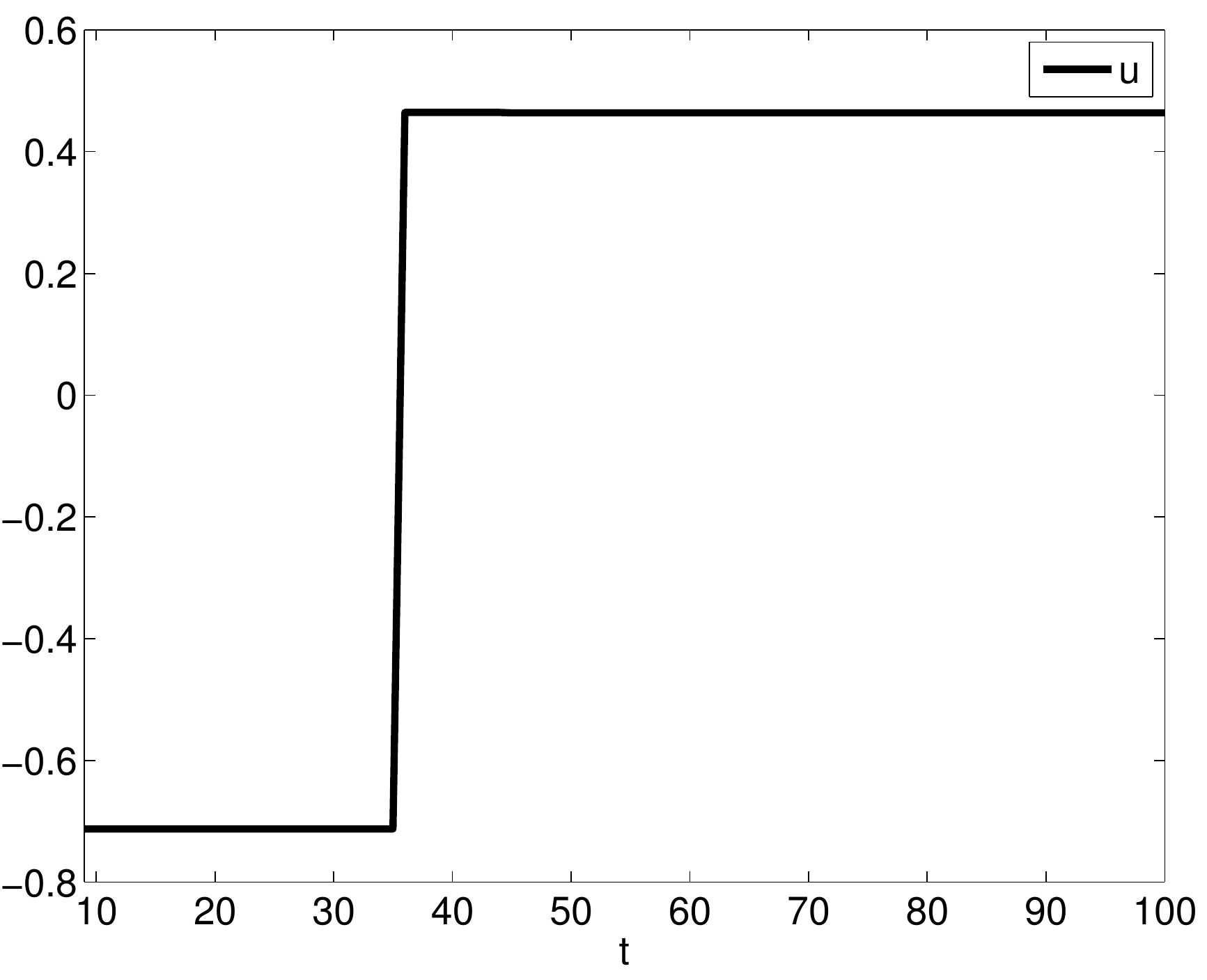}
\caption{Estimated piecewise constant input to the toaster power
  measurements seen in Figure \ref{fig:toaster1}.}\label{fig:toaster1u}
\end{figure}

\begin{figure}[h!]\centering
\includegraphics[width=0.9\columnwidth]{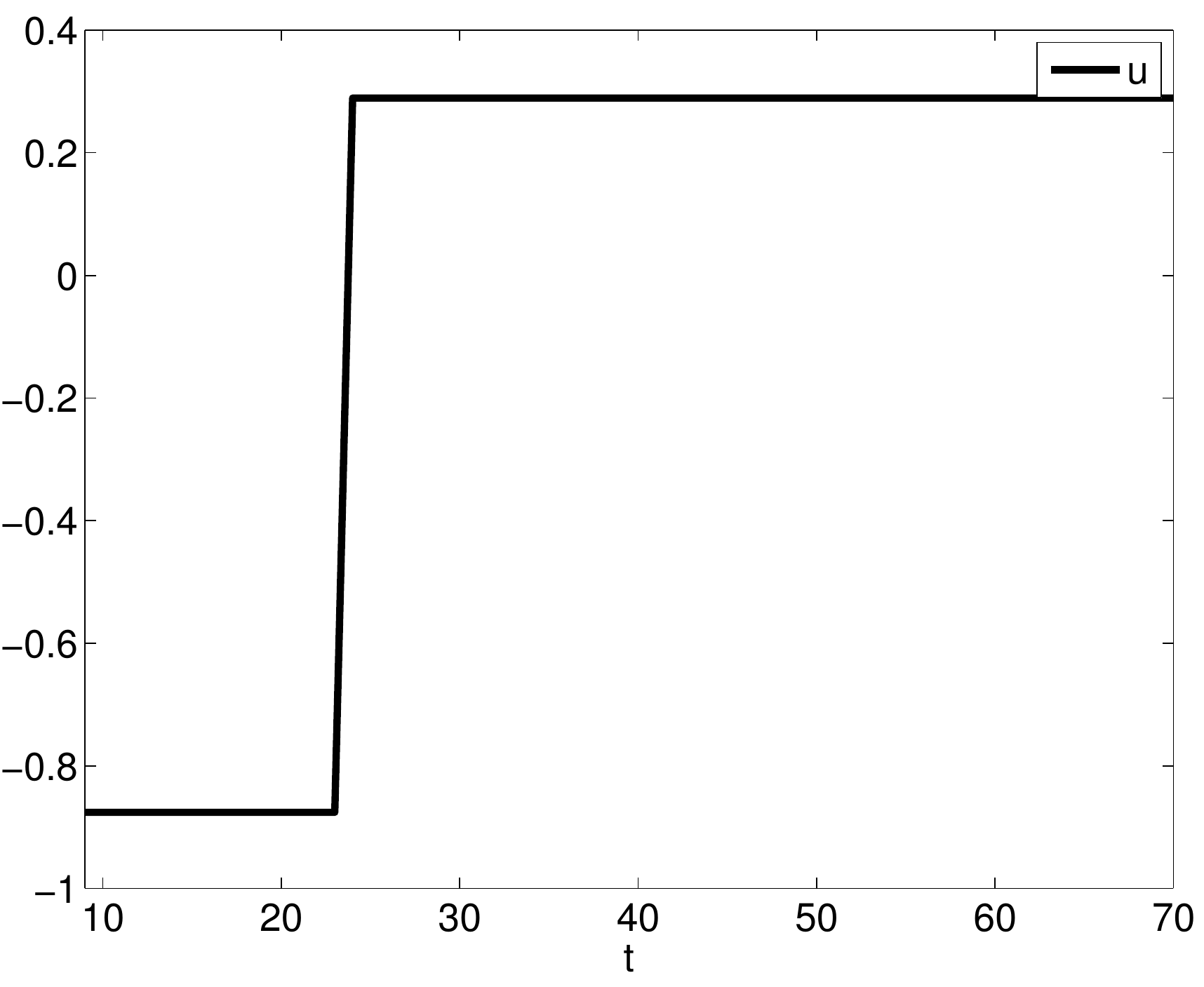}
\caption{Estimated piecewise constant input to the toaster power
  measurements seen in Figure \ref{fig:toaster2}.}\label{fig:toaster2u}
\end{figure}

\section{Conclusion}

This paper presented a novel framework for BSI of ARX model with piecewise
constant inputs. The framework uses the fact that the problem can be
rewritten as a rank minimization problem. A convex relaxation is
presented to approximate the sought ARX parameters and the unknown
inputs.

%It is common practice in sparse estimation to enhance  sparsity by
%an iterative re-weighting procedure, see for instance
%\cite{candes08}. 

%%%%%%%%%%%%%%%%%%%%%%%%%%%%%%%%%%%%%%%%%%%%%%%%%%%%%%%%%%%%%%%%%%%%%%%%%%%%%%%%
%\section{ACKNOWLEDGMENTS}

%%%%%%%%%%%%%%%%%%%%%%%%%%%%%%%%%%%%%%%%%%%%%%%%%%%%%%%%%%%%%%%%%%%%%%%%%%%%%%%%
\bibliographystyle{plain}
\bibliography{refHO,lily_refs}

\begin{thebibliography}{10}

\bibitem{abed:1997:fi}
K.~Abed-Meraim, J.-F. Cardoso, A.Y. Gorokhov, P.~Loubaton, and E.~Moulines.
\newblock On subspace methods for blind identification of single-input
  multiple-output {FIR} systems.
\newblock {\em IEEE Transactions on Signal Processing}, 45(1):42--55, Jan.

\bibitem{abed1997:jk}
K.~Abed-Meraim, W.~Qiu, and Y.~Hua.
\newblock Blind system identification.
\newblock {\em Proceedings of the IEEE}, 85(8):1310--1322, 1997.

\bibitem{Ahmed:12}
A.~Ahmed, B.~Recht, and J.~Romberg.
\newblock Blind deconvolution using convex programming.
\newblock {\em CoRR}, abs/1211.5608, 2012.

\bibitem{berinde:08}
R.~Berinde, A.~Gilbert, P.~Indyk, H.~Karloff, and M.~Strauss.
\newblock Combining geometry and combinatorics: A unified approach to sparse
  signal recovery.
\newblock In {\em Communication, Control, and Computing, 2008 46th Annual
  Allerton Conference on}, pages 798--805, September 2008.

\bibitem{bert:97}
D.~P. Bertsekas and J.~N. Tsitsiklis.
\newblock {\em Parallel and Distributed Computation: Numerical Methods}.
\newblock Athena Scientific, 1997.

\bibitem{boyd:11}
S.~Boyd, N.~Parikh, E.~Chu, B.~Peleato, and J.~Eckstein.
\newblock Distributed optimization and statistical learning via the alternating
  direction method of multipliers.
\newblock {\em Foundations and Trends in Machine Learning}, 2011.

\bibitem{bruckstein:09}
A.~Bruckstein, D.~Donoho, and M.~Elad.
\newblock From sparse solutions of systems of equations to sparse modeling of
  signals and images.
\newblock {\em SIAM Review}, 51(1):34--81, 2009.

\bibitem{Candes_2008}
E.~Cand{\`e}s.
\newblock The restricted isometry property and its implications for compressed
  sensing.
\newblock {\em Comptes Rendus Mathematique}, 346(9--10):589--592, 2008.

\bibitem{Candes:06}
E.~Cand{\`e}s, J.~Romberg, and T.~Tao.
\newblock Robust uncertainty principles: Exact signal reconstruction from
  highly incomplete frequency information.
\newblock {\em IEEE Transactions on Information Theory}, 52:489--509, February
  2006.

\bibitem{Candes:11}
E.~Cand{\`e}s, T.~Strohmer, and V.~Voroninski.
\newblock {PhaseLift}: Exact and stable signal recovery from magnitude
  measurements via convex programming.
\newblock Technical Report arXiv:1109.4499, Stanford University, September
  2011.

\bibitem{Candes:2010}
E.~J. Cand{\`e}s and Y.~Plan.
\newblock Tight oracle bounds for low-rank matrix recovery from a minimal
  number of random measurements.
\newblock {\em CoRR}, abs/1001.0339, 2010.

\bibitem{caron:2004:ji}
J.~N. Caron.
\newblock Efficient blind deconvolution of audio-frequency signal.
\newblock {\em Journal of the Acoustical Society of America}, 116(1):373--378,
  2004.

\bibitem{Chai:10}
A.~Chai, M.~Moscoso, and G.~Papanicolaou.
\newblock Array imaging using intensity-only measurements.
\newblock Technical report, Stanford University, 2010.

\bibitem{Ohlssonetal:13e}
R.~Dong, L.~Ratliff, H.~Ohlsson, and S.~S. Shankar.
\newblock A dynamical systems approach to energy disaggregation.
\newblock In {\em Proceedings of the 52th IEEE Conference on Decision and
  Control}, Florence, Italy, December 2013.
\newblock Submitted to.

\bibitem{Donoho:06}
D.~Donoho.
\newblock Compressed sensing.
\newblock {\em IEEE Transactions on Information Theory}, 52(4):1289--1306,
  April 2006.

\bibitem{Fazel01arank}
M.~Fazel, H.~Hindi, and S.~P. Boyd.
\newblock A rank minimization heuristic with application to minimum order
  system approximation.
\newblock In {\em Proceedings of the 2001 American Control Conference}, pages
  4734--4739, 2001.

\bibitem{gesbert1997:lk}
D.~Gesbert, P.~Duhamel, and S.~Mayrargue.
\newblock On-line blind multichannel equalization based on mutually referenced
  filters.
\newblock {\em IEEE Transactions on Signal Processing}, 45(9):2307--2317, 1997.

\bibitem{Goemans:1995}
M.~X. Goemans and D.~P. Williamson.
\newblock Improved approximation algorithms for maximum cut and satisfiability
  problems using semidefinite programming.
\newblock {\em J. ACM}, 42(6):1115--1145, November 1995.

\bibitem{cvx2}
M.~Grant and S.~Boyd.
\newblock Graph implementations for nonsmooth convex programs.
\newblock In V.~D. Blondel, S.~Boyd, and H.~Kimura, editors, {\em Recent
  Advances in Learning and Control}, Lecture Notes in Control and Information
  Sciences, pages 95--110. Springer-Verlag, 2008.
\newblock \url{http://stanford.edu/~boyd/graph_dcp.html}.

\bibitem{cvx1}
M.~Grant and S.~Boyd.
\newblock {CVX}: Matlab software for disciplined convex programming, version
  1.21.
\newblock \url{http://cvxr.com/cvx}, August 2010.

\bibitem{hua1996:fa}
Y.~Hua.
\newblock Fast maximum likelihood for blind identification of multiple {FIR}
  channels.
\newblock {\em IEEE Transactions on Signal Processing}, 44(3):661--672, 1996.

\bibitem{Kimetal:09}
S.-J. Kim, K.~Koh, S.~Boyd, and D.~Gorinevsky.
\newblock $\ell_1$ trend filtering.
\newblock {\em SIAM Review}, 51(2):339--360, 2009.

\bibitem{li1993:hy}
T.-H. Li.
\newblock Blind deconvolution of discrete-valued signals.
\newblock In {\em Conference Record of The Twenty-Seventh Asilomar Conference
  on Signals, Systems and Computers}, pages 1240--1244. IEEE, 1993.

\bibitem{Ljung:99}
L.~Ljung.
\newblock {\em System Identification --- Theory for the User}.
\newblock Prentice-Hall, Upper Saddle River, N.J., 2nd edition, 1999.

\bibitem{Yalmip}
J.~L{\"{o}}fberg.
\newblock {YALMIP}: {A} toolbox for modeling and optimization in {MATLAB}.
\newblock In {\em Proceedings of the CACSD Conference}, Taipei, Taiwan, 2004.

\bibitem{Lovász91}
L.~Lov\'{a}sz and A.~Schrijver.
\newblock Cones of matrices and set-functions and 0-1 optimization.
\newblock {\em SIAM Journal on Optimization}, 1:166--190, 1991.

\bibitem{Makhoul:1975:lj}
J.~I. Makhoul.
\newblock Linear prediction: A tutorial review.
\newblock {\em Proceedings of the IEEE}, 63(4):561--580, April 1975.

\bibitem{Nesterov98}
Y.~Nesterov.
\newblock Semidefinite relaxation and nonconvex quadratic optimization.
\newblock {\em Optimization Methods \& Software}, 9:141--160, 1998.

\bibitem{OhlssonLB:10}
H.~Ohlsson, L.~Ljung, and S.~Boyd.
\newblock Segmentation of {ARX}-models using sum-of-norms regularization.
\newblock {\em Automatica}, 46(6):1107--1111, 2010.

\bibitem{ohlsson:13}
H.~Ohlsson, A.~Y. Yang, R.~Dong, M.~Verhaegen, and S.~S. Sastry.
\newblock Quadratic basis pursuit.
\newblock {\em CoRR}, abs/1301.7002, 2013.

\bibitem{shor87}
N.Z. Shor.
\newblock Quadratic optimization problems.
\newblock {\em Soviet Journal of Computer and Systems Sciences}, 25:1--11,
  1987.

\bibitem{Talwar:1994:kl}
S.~Talwar, M.~Viberg, and A.~Paulraj.
\newblock Blind estimation of multiple co-channel digital signals using an
  antenna array.
\newblock {\em IEEE Signal Processing Letters}, 1(2):29--31, 1994.

\bibitem{Tibsharami:96}
R.~Tibsharani.
\newblock Regression shrinkage and selection via the lasso.
\newblock {\em Journal of Royal Statistical Society B (Methodological)},
  58(1):267--288, 1996.

\end{thebibliography}
\end{document}